\documentclass[11pt]{article}
\usepackage[margin=1in]{geometry}

\usepackage{amsmath,amssymb,amsthm}
\usepackage{bm}
\usepackage{tikz}
\usetikzlibrary{arrows.meta}
\usepackage{subcaption}
\usepackage{natbib}
\usepackage{microtype}
\usepackage[hidelinks]{hyperref}
\allowdisplaybreaks

\newtheorem{theorem}{Theorem}
\newtheorem{lemma}{Lemma}
\newtheorem{corollary}{Corollary}
\newtheorem{definition}{Definition}
\newtheorem{remark}{Remark}

\newcommand{\size}[1]{\ensuremath{|#1|}}
\newcommand{\ceil}[1]{\ensuremath{\lceil#1\rceil}}
\newcommand{\floor}[1]{\ensuremath{\lfloor#1\rfloor}}
\newcommand{\lra}[1]{\ensuremath{(#1)}}
\newcommand{\lrc}[1]{\ensuremath{\{#1\}}}
\newcommand{\lrC}[1]{\ensuremath{\left\{#1\right\}}}
\newcommand{\EE}[1]{\ensuremath{\mathbb{E}[#1]}}
\def\O{\mathcal{O}}
\def\P{\mathcal{P}}
\newcommand{\spv}{\mathrm{sp}}
\newcommand{\unv}{\mathrm{un}}
\newcommand{\Cs}{\pi^*_{\spv}}
\newcommand{\Cu}{\pi^*_{\unv}}
\newcommand{\OPT}{\operatorname{opt}}
\newcommand{\Ops}{\OPT_{\spv}}
\newcommand{\Opu}{\OPT_{\unv}}
\newcommand{\cost}{\operatorname{cost}}
\newcommand{\ind}{\mathbf{1}}

\title{Tight Approximations for Fixed-Order Vehicle Routing with Stochastic Demands}
\author{Jingyang Zhao\thanks{Kyung Hee University. Email: \texttt{jingyangzhao1020@gmail.com}.}
\and
Mingyu Xiao\thanks{University of Electronic Science and Technology of China. Email: \texttt{myxiao@gmail.com}.}
}
\date{}

\begin{document}
\maketitle
\begin{abstract}
We study the vehicle routing problem (VRP) with stochastic demands, where each client has an independent demand with a known distribution, and its realized demand is observed only on the first visit.
We consider a priori and restocking policies that fix the client order before observing demands and serve each client's entire demand before first visiting the next.
Gupta, Nagarajan, and Ravi (Oper. Res., 2012) developed randomized a priori $(\alpha+1)$-approximation and $(\alpha+2)$-approximation algorithms for splittable and unsplittable deliveries, respectively, where $\alpha$ is the approximation ratio of a metric traveling salesman problem (TSP) algorithm.
We improve these approximation ratios to $2$ and $3$ by using their replenishment rules with a TSP tour constructed via matroid parity.
We also obtain deterministic a priori $(2+\varepsilon)$-approximation and $(3+\varepsilon)$-approximation algorithms by compressing the distributions of cumulative demands to select an initial stock in polynomial time.
All guarantees are relative to the general adaptive optimum.
We then establish matching lower bounds of $2$ and $3$ that hold even for the more general class of randomized restocking policies.
Finally, we apply our matroid-parity construction to improve approximation ratios for several VRP variants with deterministic demands.
\end{abstract}


\section{Introduction}\label{sec:intro}
In vehicle routing problems (VRPs), the goal is to serve clients at minimum cost.
In the classical capacitated VRP (CVRP), demands are known before routes are constructed.
In applications such as fuel delivery, waste collection, and vending-machine replenishment, demands may be uncertain when a vehicle leaves the depot~\citep{journal:gendreau:EJOR:96,journal:florio:EJOR:23}.
The VRP with stochastic demands (VRPSD) models this uncertainty: the vehicle observes each client's demand on its first visit~\citep{journal:laporte:OR:02,journal:florio:TS:20}.
The vehicle must decide which client to visit next and when to return to the depot for replenishment.

For the VRPSD, the input consists of a metric space on a depot $r$ and a set $V$ of $n$ clients, together with a vehicle of capacity $Q\in\mathbb Z_{>0}$.
Each client $v$ has an independent demand $D_v\in\lrc{0,1,\ldots,Q}$, whose distribution is specified by its support values and probabilities. The realized demand is observed only when the vehicle first visits the client.
The vehicle starts and ends at the depot, may replenish there any number of times, and must serve all demands while minimizing expected travel cost.
In the \emph{splittable} variant, a client's demand may be served in several deliveries.
In the \emph{unsplittable} variant, it must be served in one delivery.

The literature distinguishes \emph{a priori}, \emph{restocking}, and \emph{reoptimization} approaches~\citep{journal:gendreau:TS:16}.
An a priori policy fixes the client order before observing demands and uses a prescribed shortage-response rule.
A restocking policy also fixes the order but may also choose preventive replenishments.
After fully serving a client, it chooses whether to visit the next client directly or replenish first, using the remaining stock and demand distributions.
A reoptimization policy allows the remaining order to change as demands are observed.

We study a priori and restocking policies, and compare their expected costs with the general adaptive optimum, which may change the client order or defer service after observing a demand.
For randomized policies, the expectation includes both demands and the algorithm's random choices.

\citet{journal:gupta:OR:12}, hereafter GNR, developed randomized a priori $(\alpha+1)$-approximation and $(\alpha+2)$-approximation algorithms for splittable and unsplittable deliveries, respectively, where $\alpha$ is the approximation ratio of the metric traveling salesman problem (TSP) algorithm used to construct the order.
Their algorithms construct the order and use a random initial stock with a prescribed shortage response.
These bounds hold for every fixed demand vector, with expectation taken over the initial stock.
\citet{conference:zhao:ISAAC:24} later obtained a randomized $(\alpha+1.75)$-approximation for the unsplittable VRPSD using deferred service, which allows a broader class of policies.

\subsection{Related Work}
We refer to \citet{journal:gendreau:EJOR:96,journal:gendreau:TS:16,journal:florio:EJOR:23} for surveys of stochastic vehicle routing.

\paragraph{A priori routing and restocking.}
Heuristics for a priori routing and preventive replenishments have been studied by \citet{journal:bertsimas:OR:92} and \citet{journal:yang:TS:00}, respectively.
For a fixed client order, optimal restocking decisions can be computed by stochastic dynamic programming, following the classical approach of \citet{journal:yee:NRLQ:80}. See also \citet{journal:florio:EJOR:20}.
\citet{journal:florio:EJOR:20} jointly optimized the client order and restocking decisions for a single vehicle, using a mixed-integer formulation and a heuristic.
\citet{journal:bertazzi:OR:20} bounded the worst-case benefit of restocking over a priori recourse.

\paragraph{Exact algorithms.}
\citet{journal:laporte:OR:02} developed an integer L-shaped algorithm with corrective roundtrips.
\citet{journal:christiansen:ORL:07} proposed a branch-and-price algorithm using set partitioning and dynamic programming for pricing.
For preventive replenishments, \citet{journal:louveaux:TS:18} developed a branch-and-cut algorithm based on the integer L-shaped method, and \citet{journal:florio:TS:20} proposed a branch-price-and-cut algorithm under optimal restocking.
\citet{journal:florio:EJOR:23} strengthened the latter approach by pricing elementary routes and developed a Bayesian model for replenishment under correlated demands.
\citet{journal:parada:EJOR:24} proposed a disaggregated integer L-shaped method for routing with detour-to-depot recourse.
These methods use different fleet, demand, and service assumptions.
We study polynomial-time approximation for a single vehicle.

\paragraph{Reoptimization.}
\citet{journal:secomandi:OR:01} developed a rollout policy that evaluates routing decisions using their expected future costs.
\citet{journal:secomandi:OR:09} studied partial reoptimization using a finite-horizon Markov decision process.
\citet{journal:goodson:OR:13} later developed rollout policies for multiple vehicles with route-duration limits.
\citet{journal:goodson:TS:16} incorporated fixed-route restocking values into rollout policies to evaluate routing decisions.
\citet{journal:florio:TS:22} considered limited reoptimization by allowing switches between adjacent clients.

\paragraph{Approximation algorithms for the CVRP.}
\citet{journal:haimovich:MOR:85} introduced iterated tour partitioning, starting from a TSP tour and dividing it into capacity-feasible depot tours.
Given an $\alpha$-approximate TSP tour, \citet{journal:altinkemer:TS:90} obtained an $(\alpha+1)$-approximation for the splittable CVRP, and \citet{journal:altinkemer:ORL:87} obtained an $(\alpha+2)$-approximation for the unsplittable CVRP.
Later, \citet{journal:blauth:MP:23} improved the ratios to $\alpha+1-\varepsilon_\alpha$ for splittable deliveries and $\alpha+2-2\varepsilon_\alpha$ for unsplittable deliveries, where $\varepsilon_\alpha>0$ depends only on $\alpha$.
For the unsplittable CVRP, \citet{journal:friggstad:MOR:25} proposed an $(\alpha+1+\ln2+\varepsilon)$-approximation algorithm using linear programming (LP).
They also combined their algorithm with the method of \citet{journal:blauth:MP:23} to improve the ratio by a small positive constant.
\citet{journal:zhao:arXiv:26} later obtained approximation ratios of $3.1759$ for general capacity and $3.0897$ for fixed capacity.
\citet{journal:ebert:arXiv:26} further improved the approximation ratio for general capacity to $3.159$.


\subsection{Our Results}
For the VRPSD, we first obtain randomized a priori $2$-approximation and $3$-approximation algorithms for splittable and unsplittable deliveries, respectively.
These improve the previous randomized a priori guarantees of $\alpha+1$ and $\alpha+2$~\citep{journal:gupta:OR:12}.
Our algorithms build on GNR's replenishment rules but replace the $\alpha$-approximate TSP tour with one obtained via the weighted graphic matroid parity problem (WGMP)~\citep{journal:iwata:SICOMP:22}.

We then obtain deterministic a priori $(2+\varepsilon)$-approximation and $(3+\varepsilon)$-approximation algorithms for splittable and unsplittable deliveries, respectively.
We retain the shortage-response rules of the randomized a priori algorithms.
Using the same TSP tour, a direct dynamic program selects an optimal initial stock for the splittable policy, yielding approximation ratios of $2$ and $3$ for the two variants, but its running time is pseudopolynomial in $Q$.
To obtain polynomial running time, we compress the distributions of cumulative demand and select the initial stock using the approximate expected costs.

We also establish matching lower bounds of $2$ and $3$ for a priori and restocking policies.
For every $\varepsilon>0$, there exists an instance with two-point demands on which every randomized a priori or restocking policy has an approximation ratio greater than $2-\varepsilon$ for splittable deliveries and $3-\varepsilon$ for unsplittable deliveries.
These lower bounds hold even if all demands are revealed after the client order is fixed.
Thus, even the more general restocking policies cannot achieve approximation ratios below $2$ and $3$ for splittable and unsplittable deliveries, respectively.

Finally, we show that our matroid-parity construction also improves approximation ratios for several VRP variants with deterministic demands.

\paragraph{Overview of the techniques.}
Let $c$ denote the metric cost function.
For any demand vector $\bm{q}\in\lrc{0,1,\ldots,Q}^V$, let $\Cs(\bm{q})$ be the minimum cost of a splittable routing when $\bm{q}$ is known in advance and all clients in $V$ must be visited.
Let $\ell(\bm{q})=(2/Q)\sum_{v\in V}c(r,v)q_v$ be the radial lower bound.
Our algorithm uses only the metric to construct a TSP tour $T$ satisfying $c(T)\le2\Cs(\bm{q})-\ell(\bm{q})$ for every $\bm{q}$.
To construct $T$, we first build an instance of the WGMP and compute a minimum-weight parity basis using the algorithm of \citet{journal:iwata:SICOMP:22}.
From this basis, we obtain a TSP tour by keeping certain cycles through the depot once, doubling the remaining edges, and shortcutting.
To bound its cost, we use bipartite matching to assign some clients to distinct depot tours of a splittable routing, and assign the remaining clients to tours that visit them.
These assignments yield a partition of the clients such that twice the sum of the maximum depot distances in its parts is at least $\ell(\bm{q})$.
We then use the depot tours of this partition to construct a parity basis of weight at most $2\Cs(\bm{q})-\ell(\bm{q})$.
Combining this tour bound with the expected cost bounds $c(T)+\ell(\bm{q})$ and $c(T)+2\ell(\bm{q})$ for the a priori policies~\citep{journal:gupta:OR:12}, we obtain our approximation ratios.

For deterministic a priori algorithms, we first use dynamic programming to choose the initial stock.
The direct implementation considers all $Q$ possible stock levels and is pseudopolynomial in $Q$.
To obtain polynomial running time, we express the expected extra cost using the distributions of cumulative demand modulo $Q$.
We then compress these distributions by rounding their cumulative probabilities, retaining only $\O(n/\varepsilon)$ support values each.
The resulting cost estimates have small errors for every initial stock and the same average as the true expected costs.
Thus, choosing a stock with the minimum approximate cost incurs only a small additional expected cost.
Combining this bound with our TSP tour bound yields the desired approximation ratios.


\section{Preliminaries}\label{sec:model}

\subsection{The Model and the Policy Class}
Let $G=(V\cup\lrc{r},E)$ be a complete graph, where $r$ is the depot and $V=\lrc{v_1,\ldots,v_n}$ is the client set.
A vehicle of capacity $Q\in\mathbb Z_{>0}$ starts and ends at $r$ and may replenish there any number of times.
Each client $v\in V$ has a random demand $D_v\in\lrc{0,1,\ldots,Q}$ with $\Pr(D_v>0)>0$.
The demands are mutually independent.
Let $\bm{D}=(D_v)_{v\in V}$ denote the random demand vector.
The demand of each client is observed only when the vehicle first visits it.
The input includes each client's demand distribution, specified by its finite support $\operatorname{supp}(D_v)=\lrc{d_{v,1},\ldots,d_{v,s_v}}\subseteq\lrc{0,\ldots,Q}$ and probabilities $p_{v,j}=\Pr(D_v=d_{v,j})>0$, with $\sum_{j=1}^{s_v}p_{v,j}=1$.
Let $S=\sum_{v\in V}s_v$ be the total support size.

Let $c$ be a metric cost function with $c(x,x)=0$, $c(x,y)=c(y,x)$, and $c(x,y)\le c(x,z)+c(z,y)$ for all $x,y,z\in V\cup\lrc{r}$.
For each $v\in V\cup\lrc{r}$, let $\rho_v=c(r,v)$.
A \emph{walk} in $G$ is a sequence of vertices in which consecutive vertices are joined by an edge and vertices may be repeated.
A walk is \emph{closed} if its first and last vertices are the same.
For a walk or multigraph $H$, let $E(H)$ denote its edge multiset and let $c(H)=\sum_{e\in E(H)}c(e)$ denote its total edge cost, counting multiplicities.
A \emph{depot tour} is a closed walk from $r$ with no intermediate visits to $r$.
A TSP tour on $V\cup\lrc{r}$ is a depot tour that visits every client in $V$.

A feasible policy serves all realized demands almost surely and respects the vehicle capacity.
Splittable service allows several deliveries to a client, whereas unsplittable service requires one delivery of its entire demand.
A general adaptive policy may defer service after visiting a client and choose subsequent clients based on the observed history.
Since $\Pr(D_v>0)>0$ for every $v\in V$, every feasible policy must visit every client in $V$~\citep{journal:gupta:OR:12}.

\begin{definition}[A priori and restocking policies]\label{def:class}
In both classes, the first visits to clients follow an order
fixed before observing any demand, and each client's entire
demand is served before the next client's first visit.
In either class, the policy may replenish at the depot between clients or while serving the current client.
An a priori policy uses a prescribed shortage-response rule, whereas a restocking policy may additionally choose preventive replenishments.
\end{definition}

Let $\Ops$ and $\Opu$ denote the optimal expected costs over feasible adaptive policies for splittable and unsplittable service, respectively.
For $\rho\ge1$, a $\rho$-approximation algorithm is a polynomial-time algorithm that produces a feasible policy with expected travel cost at most $\rho$ times the corresponding adaptive optimum on every instance.
The expectation is over the demands and, if the algorithm is randomized, its random choices.

\subsection{Lower Bounds on Policy Cost}
We use optimal routing costs for known demands to lower bound expected policy cost.
For $\bm{q}\in\lrc{0,1,\ldots,Q}^V$, let $\Cs(\bm{q})$ and $\Cu(\bm{q})$ be the minimum costs of splittable and unsplittable routings when the entire demand vector $\bm{q}$ is known in advance.
Both routings must visit every client in $V$, including those with $q_v=0$.
For each $\emptyset\neq B\subseteq V$, let $\tau(B)$ denote the minimum cost of a TSP tour through the depot and all clients in $B$.
The \emph{radial lower bound} is $\ell(\bm{q})=(2/Q)\sum_{v\in V}\rho_vq_v$.
We have the following result.

\begin{lemma}\label{lem:benchmarks}
$\Cs(\bm{q})\le\Cu(\bm{q})$ and, for $\nu\in\lrc{\spv,\unv}$, $\max\lrc{\tau(V),\ell(\bm{q})}\le \pi^*_\nu(\bm{q})$ and $\EE{\pi^*_\nu(\bm{D})}\le\OPT_{\nu}$.
\end{lemma}
\begin{proof}
The first inequality follows by allowing split deliveries.
The second follows from the TSP and radial lower bounds of \citet{journal:gupta:OR:12}, since every routing must visit all clients.
Finally, for any feasible adaptive policy of variant $\nu$, its realized routing has cost at least $\pi^*_\nu(\bm{D})$ almost surely, so taking expectations yields the last inequality.
\end{proof}

\section{Randomized A Priori Algorithms}\label{sec:random}
In this section, we present randomized a priori $2$-approximation and $3$-approximation algorithms for splittable and unsplittable deliveries, respectively.

\subsection{The A Priori Policies}
For any TSP tour $T=(r,v_1,\ldots,v_n,r)$, we use the replenishment rules of \citet{journal:gupta:OR:12} with a shortcut in the unsplittable case.
Both policies start with stock $U$, sampled uniformly from $\lrc{0,\ldots,Q-1}$. The policies follow the fixed order and serve each demand immediately if the remaining stock suffices. For a shortage, they use the following rules, respectively.

\paragraph{Splittable service.}
Suppose the vehicle visits a client with stock $s$ and observes demand $q>s$. It serves $s$ units of demand, makes one roundtrip to the depot to replenish to capacity $Q$, and serves the remaining demand $q-s$. Its stock is then $Q+s-q$.

\paragraph{Unsplittable service.}
If $q>s$, the vehicle makes a roundtrip to replenish to $q$ and serves the whole demand. It then visits the depot to replenish to $Q+s-q$ and travels directly to the next client, or ends at the depot after the last client. Its stock is again $Q+s-q$.

For the a priori policy on $T$, let $\cost_\nu(T,\bm{q})$ denote its travel cost on demand vector $\bm{q}$, where $\nu\in\lrc{\spv,\unv}$ denotes the delivery variant.
Unlike the unsplittable policy of \citet{journal:gupta:OR:12}, which returns to the current client after the second depot visit, our policy travels directly to the next client. By the triangle inequality, this cannot increase the cost. 
We obtain the following bounds.

\begin{lemma}[{\citet[Theorems~3.1 and~4.1]{journal:gupta:OR:12}}]\label{lem:gnr-cost}
For any fixed TSP tour $T$ and $\bm{q}\in\lrc{0,1,\ldots,Q}^V$, it holds that $\mathbb E_U[\cost_{\spv}(T,\bm{q})]=c(T)+\ell(\bm{q})$ and $\mathbb E_U[\cost_{\unv}(T,\bm{q})]\le c(T)+2\ell(\bm{q})$.
\end{lemma}
\begin{proof}
For fixed $\bm{q}$, the stock before visiting $v_i$ is $(U-\sum_{j<i}q_{v_j})\bmod Q$ in both policies. Since $U$ is uniformly distributed, a shortage occurs at $v_i$ with probability $q_{v_i}/Q$. 
Moreover, each shortage incurs an extra cost of $2\rho_{v_i}$ in the splittable policy and at most $4\rho_{v_i}$ in the unsplittable policy, where the latter bound follows from the triangle inequality. Thus, the claimed bounds follow by linearity of expectation.
\end{proof}

The algorithms of \citet{journal:gupta:OR:12} use an $\alpha$-approximate TSP tour $T$ satisfying $c(T)\le\alpha\tau(V)$.
By Lemmas~\ref{lem:benchmarks} and~\ref{lem:gnr-cost}, this yields randomized a priori $(\alpha+1)$-approximation and $(\alpha+2)$-approximation algorithms for splittable and unsplittable deliveries, respectively.

\subsection{Constructing a TSP Tour via Matroid Parity}\label{sec:geometry}
Using only the metric, we construct a TSP tour $T$ satisfying $c(T)\le2\Cs(\bm{q})-\ell(\bm{q})$ for every $\bm{q}\in\lrc{0,1,\ldots,Q}^V$.
Our algorithm first constructs an instance of the WGMP and computes a minimum-weight parity basis.
It then uses this basis to construct cycles through the depot, keep each cycle once, and double the remaining edges.
Finally, a TSP tour $T$ is obtained by shortcutting.

\subsubsection{The WGMP Instance}
We first review the WGMP~\citep{journal:iwata:SICOMP:22}.
An instance of the WGMP consists of a multigraph whose edges are partitioned into pairs called \emph{lines}, each with an assigned weight.
A \emph{parity basis} is a set of lines such that the union of their edges forms a spanning tree in each connected component of the multigraph.
The goal is to seek a parity basis of minimum total line weight.

We start with two vertex-disjoint complete graphs $G_1$ and $G_2$ with vertex labels $X\cup\lrc{r_1}$ and $X\cup\lrc{r_2}$, respectively, where $X=\lrc{v^-,v^+:v\in V}$ contains two distinct copies of each client.
An edge $r_i x$ with $x\in X$ is called a \emph{depot spoke}.
In $G_1$, we replace each depot spoke with $2n-1$ parallel edges, while $G_2$ remains unchanged.
Thus, $G_1$ has $6n^2-3n$ edges and $G_2$ has $2n^2+n$ edges.

Let $\theta:X\to V$ map each copy to its client.
In each graph, we assign an edge $xy$ on $X$ the cost $c_X(x,y)=c(\theta(x),\theta(y))$, and each depot spoke $r_i x$ the cost $\rho_{\theta(x)}$.
In particular, the edge $v^-v^+$ has cost zero, although its endpoints are distinct.
We extend the notation $c(H)$ to these graphs.
We group the edges into the following three types of lines, as illustrated in Figure~\ref{fig:parity-instance}.

\begin{figure}[t]
\centering
\caption{The three types of lines in the constructed instance.}\label{fig:parity-instance}
\vspace{-6pt}
\begin{subfigure}[t]{0.32\linewidth}
\centering
\caption{Client line}
\label{fig:parity-instance-a}
\begin{tikzpicture}[x=1.12cm,y=0.85cm,
  every node/.style={font=\footnotesize,inner sep=1pt},
  vertex/.style={circle,draw,fill=white,inner sep=0pt,minimum size=4.8pt},
  depot/.style={rectangle,draw,fill=black,inner sep=0pt,minimum size=5pt},
  edge/.style={line width=0.6pt},
  graph box/.style={draw=black!40,line width=0.35pt,rounded corners=1.5pt}]

\draw[graph box] (0.10,-0.35) rectangle (1.55,1.50);
\draw[graph box] (1.80,-0.35) rectangle (3.25,1.50);
\node at (0.825,1.10) {$G_1$};
\node at (2.525,1.10) {$G_2$};
\node[vertex,label=below:$x$] (cx1) at (0.35,0.40) {};
\node[vertex,label=below:$y$] (cy1) at (1.30,0.40) {};
\node[vertex,label=below:$x$] (cx2) at (2.05,0.40) {};
\node[vertex,label=below:$y$] (cy2) at (3.00,0.40) {};
\draw[edge] (cx1)--(cy1) (cx2)--(cy2);
\node at (1.675,-0.73) {$2c_X(x,y)$};
\end{tikzpicture}
\end{subfigure}
\begin{subfigure}[t]{0.32\linewidth}
\centering
\caption{Depot line}
\label{fig:parity-instance-b}
\begin{tikzpicture}[x=1.12cm,y=0.85cm,
  every node/.style={font=\footnotesize,inner sep=1pt},
  vertex/.style={circle,draw,fill=white,inner sep=0pt,minimum size=4.8pt},
  depot/.style={rectangle,draw,fill=black,inner sep=0pt,minimum size=5pt},
  edge/.style={line width=0.6pt},
  graph box/.style={draw=black!40,line width=0.35pt,rounded corners=1.5pt}]

\draw[graph box] (4.10,-0.35) rectangle (6.80,1.50);
\node at (4.45,1.20) {$G_1$};
\node[depot,label=above:$r_1$] (dr) at (5.45,1.05) {};
\node[vertex,label=below:$x$] (dx) at (4.70,0.15) {};
\node[vertex,label=below:$y$] (dy) at (6.20,0.15) {};
\draw[edge] (dr)--(dx) (dr)--(dy);
\node at (5.45,-0.73) {$2\rho_{\theta(x)}+2\rho_{\theta(y)}$};
\end{tikzpicture}
\end{subfigure}
\begin{subfigure}[t]{0.32\linewidth}
\centering
\caption{Marker line}
\label{fig:parity-instance-c}
\begin{tikzpicture}[x=1.12cm,y=0.85cm,
  every node/.style={font=\footnotesize,inner sep=1pt},
  vertex/.style={circle,draw,fill=white,inner sep=0pt,minimum size=4.8pt},
  depot/.style={rectangle,draw,fill=black,inner sep=0pt,minimum size=5pt},
  edge/.style={line width=0.6pt},
  graph box/.style={draw=black!40,line width=0.35pt,rounded corners=1.5pt}]

\draw[graph box] (7.85,-0.35) rectangle (10.55,1.50);
\node at (8.20,1.20) {$G_2$};
\node[depot,label=above:$r_2$] (mr) at (9.20,1.05) {};
\node[vertex,label=below:$v^-$] (mminus) at (8.45,0.15) {};
\node[vertex,label=below:$v^+$] (mplus) at (9.95,0.15) {};
\draw[edge] (mr)--(mminus) (mr)--(mplus);
\node at (9.20,-0.73) {$-2\rho_v$};
\end{tikzpicture}
\end{subfigure}
\end{figure}
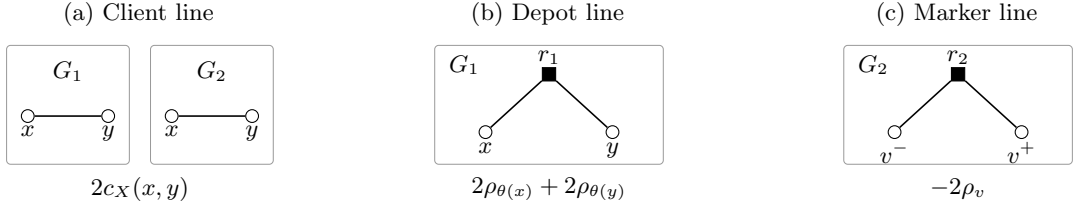

\begin{enumerate}
\item \emph{Client lines.}
For each $\lrc{x,y}\subseteq X$ with $x\ne y$, we pair the two occurrences of the edge $xy$, one in each graph, in a line of weight $2c_X(x,y)$.

\item \emph{Depot lines.}
For each $\lrc{x,y}\subseteq X$ with $x\ne y$, we pair one unused occurrence of each spoke $r_1x$
and $r_1y$ in $G_1$ in a line of weight $2\rho_{\theta(x)}+2\rho_{\theta(y)}$.

\item \emph{Marker lines.}
For each client $v$, we pair the spokes $r_2v^-,r_2v^+$ in $G_2$ in a line of weight $-2\rho_v$.
\end{enumerate}

Each edge belongs to exactly one line.
The multigraph of our WGMP instance is the disjoint union $G_1\cup G_2$, with $4n+2$ vertices, $8n^2-2n$ edges, and $4n^2-n$ lines.
To obtain a parity basis, we can select all marker lines and the $n$ depot lines pairing $v^-$ with $v^+$.
The selected edges form a star in each graph.
Since each graph has $2n+1$ vertices, every parity basis consists of $4n$ edges, or $2n$ lines.
We compute a minimum-weight parity basis for this instance using the algorithm of \citet{journal:iwata:SICOMP:22}.

In the instance, the line weights may be negative.
We can add $2\max_{v\in V}\rho_v$ to every line weight to make the weights nonnegative.
Since every parity basis has $2n$ lines, this increases every basis weight by the same amount $4n\max_{v\in V}\rho_v$ and preserves a minimum-weight parity basis.

\subsubsection{Constructing a TSP Tour from a Parity Basis}
We now construct a TSP tour from the selected parity basis.
Let $F$ be the forest with vertex set $X$ and the selected client edges.
Let $E_{\mathrm{d}}$ and $E_{\mathrm{m}}$ be the selected depot spokes in $G_1$ and $G_2$, respectively.
Let $V_{\mathrm{m}}$ be the set of clients whose marker lines are selected.
The client edges appear in both graphs.
Since the selected lines form a parity basis, $F\cup E_{\mathrm{d}}$ and $F\cup E_{\mathrm{m}}$ are spanning trees.
The selected marker lines pair the components of $F$, as shown by the following lemma.

\begin{lemma}\label{lem:pairing}
Each component of $F$ contains exactly one endpoint of a spoke in $E_{\mathrm{d}}$ and one in $E_{\mathrm{m}}$. The pairs $v^-,v^+$ for $v\in V_{\mathrm{m}}$ pair distinct components of $F$.
\end{lemma}
\begin{proof}
Since $F\cup E_{\mathrm{d}}$ and $F\cup E_{\mathrm{m}}$ are spanning trees, each component of $F$ has at least one spoke from each set.
Two spokes from the same set in one component would form a cycle with the path between their endpoints.
For each selected marker line, $v^-$ and $v^+$ lie in distinct components of $F$.
The unique $E_{\mathrm{m}}$ spoke in each component therefore pairs it with exactly one other component.
\end{proof}

\begin{lemma}\label{lem:decoding}
Given a parity basis of total weight $W$ in the constructed instance of the WGMP, we can construct a TSP tour $T$ satisfying $c(T)\le W$.
\end{lemma}
\begin{proof}
By Lemma~\ref{lem:pairing}, the components of $F$ are paired.
For each $v\in V_{\mathrm{m}}$, we join its two components by the zero-cost edge $v^-v^+$ and connect them to $r$ using their two $E_{\mathrm{d}}$ spokes, as shown in Figure~\ref{fig:parity-decoding}.
This creates a cycle $Z_v$ through $r$ and both copies of $v$. These cycles share only the depot.
The two paths on $Z_v$ from $r$ to $v^-$ and $v^+$ that avoid the zero-cost edge each have cost at least $\rho_v$ by the triangle inequality.
Thus, $c(Z_v)\ge2\rho_v$.

We keep each cycle once and double every edge outside the cycles.
The resulting multigraph is connected, spans all copies, and has even degree at every vertex.
By the line weights, it holds that $W=2c(F)+2c(E_{\mathrm{d}})-2\sum_{v\in V_{\mathrm{m}}}\rho_v$.
Since the added edges have cost zero, the multigraph has cost
\begin{equation}\label{eq:decode}
 2c(F\cup E_{\mathrm{d}})-\sum_{v\in V_{\mathrm{m}}}c(Z_v)
 \le2c(F\cup E_{\mathrm{d}})-2\sum_{v\in V_{\mathrm{m}}}\rho_v=W.
\end{equation}

We take an Eulerian tour of this multigraph, map each copy to its client, and shortcut repeated clients and intermediate depot visits.
By the triangle inequality, the resulting TSP tour $T$ has cost at most $W$.
\end{proof}

\begin{figure}[t]
\centering
\caption{The depot cycle $Z_v$ formed by the two components paired by the marker line for $v$.}\label{fig:parity-decoding}
\vspace{-6pt}
\begin{subfigure}[t]{0.49\linewidth}
\centering
\caption{The two paired components}
\label{fig:parity-decoding-a}
\begin{tikzpicture}[x=1.12cm,y=0.78cm,
  every node/.style={font=\small},
  vertex/.style={circle,draw,fill=white,inner sep=0pt,minimum size=4.8pt},
  depot/.style={rectangle,draw,fill=black,inner sep=0pt,minimum size=5pt},
  marker/.style={circle,draw,fill=black,inner sep=0pt,minimum size=5pt},
  edge/.style={line width=0.6pt},
  kept/.style={line width=1pt},
  doubled/.style={double,double distance=1.3pt,line width=0.4pt}]
\node[depot,label=above:$r$] (r) at (2.1,2.2) {};
\node[vertex] (a) at (0.9,1.15) {};
\node[vertex] (b) at (3.3,1.15) {};
\node[marker,label=below:$v^-$] (vm) at (1.3,0) {};
\node[marker,label=below:$v^+$] (vp) at (2.9,0) {};
\node[vertex] (c) at (0,0.55) {};
\node[vertex] (d) at (4.2,0.55) {};
\draw[edge] (r)--node[above left,inner sep=1pt] {$E_{\mathrm{d}}$} (a);
\draw[edge] (r)--node[above right,inner sep=1pt] {$E_{\mathrm{d}}$} (b);
\draw[edge] (a)--(vm) (a)--(c) (b)--(vp) (b)--(d);
\end{tikzpicture}
\end{subfigure}
\begin{subfigure}[t]{0.49\linewidth}
\centering
\caption{Keep the cycle and double other edges}
\label{fig:parity-decoding-b}
\begin{tikzpicture}[x=1.12cm,y=0.78cm,
  every node/.style={font=\small},
  vertex/.style={circle,draw,fill=white,inner sep=0pt,minimum size=4.8pt},
  depot/.style={rectangle,draw,fill=black,inner sep=0pt,minimum size=5pt},
  marker/.style={circle,draw,fill=black,inner sep=0pt,minimum size=5pt},
  edge/.style={line width=0.6pt},
  kept/.style={line width=1pt},
  doubled/.style={double,double distance=1.3pt,line width=0.4pt}]
\node[depot,label=above:$r$] (r2) at (2.1,2.2) {};
\node[vertex] (a2) at (0.9,1.15) {};
\node[vertex] (b2) at (3.3,1.15) {};
\node[marker,label=below:$v^-$] (vm2) at (1.3,0) {};
\node[marker,label=below:$v^+$] (vp2) at (2.9,0) {};
\node[vertex] (c2) at (0,0.55) {};
\node[vertex] (d2) at (4.2,0.55) {};
\draw[kept] (r2)--(a2)--(vm2);
\draw[kept] (r2)--(b2)--(vp2);
\draw[kept,dashed] (vm2)--node[below,inner sep=2pt] {$0$} (vp2);
\draw[doubled] (a2)--(c2) (b2)--(d2);
\node at (2.1,1.02) {$Z_v$};
\end{tikzpicture}
\end{subfigure}
\end{figure}
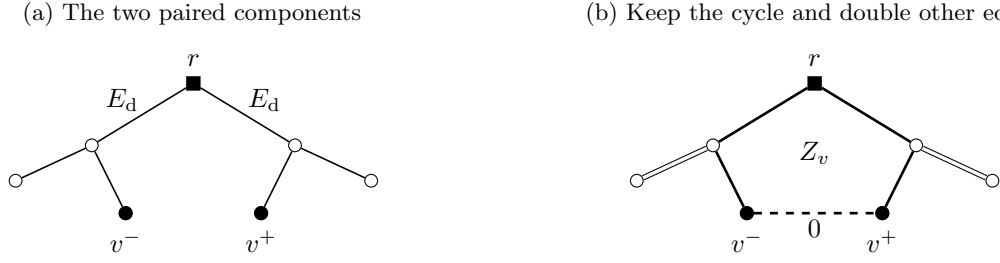

\subsection{Analysis}\label{sec:guarantees}
For each nonempty $B\subseteq V$, let $\rho_{\max}(B)=\max_{v\in B}\rho_v$, and define
\begin{equation}\label{eq:partition}
 \pi_{\mathrm{part}}=\min_{\P\text{ partition of }V}\sum_{B\in\P}\lra{2\tau(B)-2\rho_{\max}(B)}.
\end{equation}
This value depends only on the metric and imposes no capacity constraints on the depot tours.
We first show that $\pi_{\mathrm{part}}\le2\Cs(\bm{q})-\ell(\bm{q})$ for every demand vector $\bm{q}$ using a partition obtained from a splittable routing.
We then show that the minimum parity basis weight of our WGMP instance is at most $\pi_{\mathrm{part}}$.

\paragraph{From splittable routings to a client partition.}
We use a matching argument to partition the clients
while preserving the routing cost and the radial bound.

\begin{lemma}\label{lem:allocation}
For any $\bm{q}\in\lrc{0,1,\ldots,Q}^V$ and every splittable routing of cost $\pi$ that serves $\bm{q}$ and visits every client in $V$, there exist a partition $\P$ of $V$ and, for each $B\in\P$, a depot tour $T_B$ visiting exactly the clients of $B$ once, such that $\sum_{B\in\P}c(T_B)\le \pi$ and $2\sum_{B\in\P}\rho_{\max}(B)\ge \ell(\bm{q})$.
\end{lemma}
\begin{proof}
Index the depot tours of the routing by $j$, and let $x_{j,v}$ be the amount of demand served at $v$ on tour $j$.
Construct a bipartite graph with the tours on one side and the clients on the other, and include an edge $jv$ of weight $\rho_v$ whenever tour $j$ visits $v$.
Since $\sum_vx_{j,v}\le Q$ and $\sum_jx_{j,v}=q_v\le Q$, the values $x_{j,v}/Q$ define a fractional matching of weight $\sum_{j,v}\rho_vx_{j,v}/Q=\ell(\bm{q})/2$.
By integrality of bipartite matching, we can choose a matching of weight at least $\ell(\bm{q})/2$.
Note that this matching may leave clients and depot tours unmatched.

Then, assign each matched client to its matched tour and every other client to any tour visiting it, including clients with $q_v=0$.
Shortcut each tour to its assigned clients and remove tours with no assigned clients.
The assigned client sets form a partition $\P$ of $V$, and the resulting tours satisfy $\sum_{B\in\P}c(T_B)\le \pi$.
Since each resulting tour contains at most one matched client, we obtain $2\sum_{B\in\P}\rho_{\max}(B)\ge2\sum_{v\text{ matched}}\rho_v\ge\ell(\bm{q})$.
\end{proof}

For the partition and tours in Lemma~\ref{lem:allocation}, it holds that $\tau(B)\le c(T_B)$ for each $B\in\P$.
Hence, by \eqref{eq:partition},
\begin{equation}\label{eq:partition-transfer}
 \pi_{\mathrm{part}}\le\sum_{B\in\P}\lra{2c(T_B)-2\rho_{\max}(B)}\le2\pi-\ell(\bm{q}).
\end{equation}

\paragraph{Constructing a parity basis from depot tours.}
We next construct a parity basis using depot tours on a client partition and bound the minimum parity basis weight.
\begin{lemma}\label{lem:embedding}
For every partition $\P$ of $V$ and depot tours $T_B$ visiting exactly the clients of $B$ once for each $B\in\P$, the constructed instance of the WGMP has a parity basis of weight $\sum_{B\in\P}\lra{2c(T_B)-2\rho_{\max}(B)}$.
\end{lemma}
\begin{proof}
For each $B\in\P$, choose $p_B\in B$ with $\rho_{p_B}=\rho_{\max}(B)$ and split $T_B$ into two depot-to-$p_B$ paths.
Use $p_B^-$ and $p_B^+$ as their respective endpoints, and replace every other client $v$ by consecutive copies $v^-,v^+$ joined by a zero-cost edge.
Let $E_{\mathrm{d}}$ be their depot spokes in $G_1$, and let $F$ be the forest on $X$ formed by their remaining edges.
The paths cover every copy once and share only $r_1$, so $F\cup E_{\mathrm{d}}$ is a spanning tree of cost $\sum_{B\in\P}c(T_B)$.

Each component of $F$ is a path with an endpoint $p_B^-$ or $p_B^+$ for some $B\in\P$.
We connect these endpoints to $r_2$ using the marker lines for $p_B$, and let $E_{\mathrm{m}}$ be their spokes in $G_2$.
Then, $F\cup E_{\mathrm{m}}$ is also a spanning tree.
The client lines corresponding to $F$, the depot line for each block, and the marker line for each $p_B$ therefore form a parity basis.
Its total weight is $2c(F)+2c(E_{\mathrm{d}})-2\sum_{B\in\P}\rho_{p_B}=\sum_{B\in\P}\lra{2c(T_B)-2\rho_{\max}(B)}$.
For a singleton block, the two paths consist of the distinct spokes to $p_B^-$ and $p_B^+$, so the same construction applies.
\end{proof}

\paragraph{The tour cost and approximation guarantees.}
We first analyze the performance of $T$.

\begin{lemma}\label{lem:geometry}
The TSP algorithm constructs in $\O(n^7)$ time a TSP tour $T$ satisfying $c(T)\le \pi_{\mathrm{part}}\le2\Cs(\bm{q})-\ell(\bm{q})$ for every $\bm{q}\in\lrc{0,1,\ldots,Q}^V$.
\end{lemma}
\begin{proof}
By Lemma~\ref{lem:embedding}, the minimum parity basis weight is at most $\pi_{\mathrm{part}}$.
Then, by Lemma~\ref{lem:decoding}, it holds that $c(T)\le \pi_{\mathrm{part}}$.
For any splittable routing of cost $\pi$ that serves $\bm{q}$ and visits every client in $V$, \eqref{eq:partition-transfer} implies $\pi_{\mathrm{part}}\le2\pi-\ell(\bm{q})$.
In particular, we obtain $\pi_{\mathrm{part}}\le2\Cs(\bm{q})-\ell(\bm{q})$.

We now analyze the running time of the TSP algorithm.
The constructed WGMP instance has $\O(n)$ vertices, $\O(n^2)$ edges, and $\O(n^2)$ lines.
Constructing its graphs and lines takes $\O(n^2)$ time.
For this graphic matroid, since it can be represented over the field $\lrc{0,1}$ with arithmetic modulo $2$ and each field operation takes $\O(1)$ time, the running time of the deterministic algorithm of \citet[Theorem~11.1]{journal:iwata:SICOMP:22} is cubic in the number of edges and linear in the size of a basis.
Since every basis has $\O(n)$ edges, computing a minimum-weight parity basis takes $\O((n^2)^3n)=\O(n^7)$ time.
Constructing the TSP tour from the $\O(n)$ selected edges takes $\O(n)$ time.
Thus, the total running time is $\O(n^7)$.
\end{proof}

We then apply the a priori policies to obtain the following approximation guarantees.

\begin{theorem}\label{thm:random}
The VRPSD admits randomized a priori $2$-approximation and $3$-approximation algorithms for splittable and unsplittable deliveries, respectively.
\end{theorem}
\begin{proof}
By Lemmas~\ref{lem:benchmarks},~\ref{lem:gnr-cost}, and~\ref{lem:geometry}, for every demand vector $\bm{q}$, it holds that $\mathbb E_U[\cost_{\spv}(T,\bm{q})]\le2\Cs(\bm{q})$ and $\mathbb E_U[\cost_{\unv}(T,\bm{q})]\le2\Cs(\bm{q})+\ell(\bm{q})\le3\Cs(\bm{q})\le3\Cu(\bm{q})$.
Taking expectations over the demands and applying Lemma~\ref{lem:benchmarks}, we obtain the claimed approximation ratios.

For the running time, both algorithms construct $T$ in $\O(n^7)$ time and sample $U$ in expected $\O(\log Q)$ time.
Each client requires $\O(1)$ stock operations and at most two replenishments, so serving all clients takes $\O(n)$ time.
Thus, the total expected running time is $\O(n^7+\log Q)$.
\end{proof}

As noted by \citet{journal:gupta:OR:12}, these randomized a priori algorithms require no knowledge of the demand distributions under our assumption that $\Pr(D_v>0)>0$ for every $v\in V$.

\section{Deterministic A Priori Algorithms}\label{sec:known}
In this section, we obtain deterministic a priori $(2+\varepsilon)$-approximation and $(3+\varepsilon)$-approximation algorithms for splittable and unsplittable deliveries, respectively.
We retain the replenishment rules from Section~\ref{sec:random} but choose the initial stock deterministically.
We first evaluate the expected cost for every initial stock by dynamic programming, obtaining approximation ratios of $2$ and $3$ in pseudopolynomial time.
We then compress the distributions of cumulative demand modulo $Q$ to obtain polynomial running time.

Fix a TSP tour $T=(r=v_0,v_1,\ldots,v_n,v_{n+1}=r)$.
With initial stock $u\in\lrc{0,\ldots,Q-1}$, both policies have the same stock upon arrival at each client, so they have the same shortages.
Each shortage at $v_i$ incurs an extra cost of $2\rho_{v_i}$ for splittable deliveries and at most $4\rho_{v_i}$ for unsplittable deliveries.

\subsection{Choosing the Initial Stock}
For $s\in\lrc{0,\ldots,Q-1}$, let $R_i(s)$ be the expected extra cost of serving the remaining clients after serving $v_i$ with remaining stock $s$, under the splittable a priori policy.
Here, $R_0(s)$ is the expected extra cost before visiting the first client, and $R_n(s)=0$.
Since the stock after serving demand $q$ is $(s-q)\bmod Q$, for $i=n,\ldots,1$,
\begin{equation}\label{eq:initial-dp}
 R_{i-1}(s)=\EE{2\rho_{v_i}\ind_{\lrc{D_{v_i}>s}}+R_i((s-D_{v_i})\bmod Q)}.
\end{equation}
The recurrence follows from the prescribed shortage-response rule and the independence of the demands.
With initial stock $u$, the expected travel cost is $c(T)+R_0(u)$ for splittable deliveries and at most $c(T)+2R_0(u)$ for unsplittable deliveries.
Thus, we can compute the smallest initial stock $u^*$ minimizing $R_0(u)$.

\begin{theorem}\label{thm:exact-priori}
The VRPSD admits deterministic a priori algorithms with approximation ratios $2$ and $3$ for splittable and unsplittable deliveries, respectively, with running time $\O(n^7+QS)$.
\end{theorem}
\begin{proof}
Use the TSP tour $T$ from Lemma~\ref{lem:geometry} and initial stock $u^*$ in both policies.
By Lemma~\ref{lem:gnr-cost}, it holds that $R_0(u^*)\le (1/Q)\sum_{u=0}^{Q-1}R_0(u)=\EE{\ell(\bm{D})}$.
Then, by Lemma~\ref{lem:benchmarks}, we obtain $c(T)+R_0(u^*)\le2\EE{\Cs(\bm{D})}\le2\Ops$ and $c(T)+2R_0(u^*)\le2\EE{\Cs(\bm{D})}+\EE{\ell(\bm{D})}\le3\EE{\Cu(\bm{D})}\le3\Opu$.
For the running time, at client $v_i$, evaluating \eqref{eq:initial-dp} takes $\O(Qs_{v_i})$ time.
Summing over all clients and selecting a minimizing initial stock takes $\O(QS)$ time.
Thus, including the tour construction in Lemma~\ref{lem:geometry} yields the total running time $\O(n^7+QS)$.
\end{proof}

To avoid the pseudopolynomial dependence on $Q$, we next approximate $R_0$ by compressing the distributions of cumulative demands.

\subsection{Compressing the Distributions}
For $0\le i\le n$, let $X_i=\sum_{j=1}^iD_{v_j}$ be the total demand of the first $i$ clients, and let $Y_i=X_i\bmod Q$.
The number of shortages among these clients can be expressed using $X_i$ and $Y_i$.
For an initial stock $u\in\lrc{0,\ldots,Q-1}$, define $h_u(z)=\ind_{\lrc{z>u}}-z/Q$ for $z\in\lrc{0,\ldots,Q-1}$.
For $1\le i\le n$, let $b_i=\rho_{v_i}-\rho_{v_{i+1}}$, where $\rho_{v_{n+1}}=0$.

\begin{lemma}\label{lem:initial-cost}
For every initial stock $u$, it holds that $R_0(u)=\EE{\ell(\bm{D})}+2\sum_{i=1}^n b_i\EE{h_u(Y_i)}$.
\end{lemma}
\begin{proof}
Let $N_i(u)$ be the number of shortages among the first $i$ clients, with $N_0(u)=0$.
After serving the first $i$ clients, the remaining stock is $u+QN_i(u)-X_i\in\lrc{0,\ldots,Q-1}$.
Consequently, $N_i(u)=\floor{X_i/Q}+\ind_{\lrc{Y_i>u}}=X_i/Q+h_u(Y_i)$.
Since $N_i(u)-N_{i-1}(u)$ indicates whether a shortage occurs at $v_i$, we obtain
\[
R_0(u)
=2\sum_{i=1}^n\rho_{v_i}\EE{N_i(u)-N_{i-1}(u)}
=2\sum_{i=1}^n b_i\EE{N_i(u)}.
\]
The terms involving $X_i/Q$ satisfy
\[
\frac2Q\sum_{i=1}^n b_i\EE{X_i}
=\frac2Q\sum_{j=1}^n\EE{D_{v_j}}
  \sum_{i=j}^n b_i
=\frac2Q\sum_{j=1}^n\rho_{v_j}\EE{D_{v_j}}
=\EE{\ell(\bm D)},
\]
where $\sum_{i=j}^n b_i=\rho_{v_j}$.
The remaining terms are $2\sum_i b_i\EE{h_u(Y_i)}$, which proves the claimed equality.
\end{proof}

Let $\mu_i(z)=\Pr(Y_i=z)$, with $\mu_0(0)=1$.
By Lemma~\ref{lem:initial-cost}, we only need to approximate the expected values of $h_u$ under these distributions.
However, each $\mu_i$ may have up to $Q$ support values.
We therefore compress the distributions while keeping these expected values nearly unchanged.

\begin{lemma}\label{lem:distribution-compression}
For any distribution $\mu$ on $\lrc{0,\ldots,Q-1}$ and integer $K\ge1$, a distribution $\hat\mu$ with at most $K$ support values can be computed such that $|\sum_z(\mu(z)-\hat\mu(z))h_u(z)|\le1/K$ for every $u$.
\end{lemma}
\begin{proof}
Let $\Gamma(t)=\sum_{z\le t}\mu(z)$ and $\Gamma(-1)=0$.
At each support value $z$, assign probability $\hat\mu(z)=(\floor{K\Gamma(z)}-\floor{K\Gamma(z-1)})/K$, keeping only positive probabilities.
These probabilities sum to $1$, and each positive one is at least $1/K$.
Thus, $\hat\mu$ has at most $K$ support values.

Let $\hat\Gamma(t)=\sum_{z\le t}\hat\mu(z)$.
Then, $\hat\Gamma(t)=\floor{K\Gamma(t)}/K$, so $0\le\Gamma(t)-\hat\Gamma(t)<1/K$.
Note that $\sum_z z\mu(z)=\sum_{t=0}^{Q-1}(1-\Gamma(t))$.
Then, $\sum_z\mu(z)h_u(z)=(1/Q)\sum_{t=0}^{Q-1}\Gamma(t)-\Gamma(u)$, and the same equality holds for $\hat\mu$ and $\hat\Gamma$.
Hence,
\[
 \sum_z(\mu(z)-\hat\mu(z))h_u(z)
 =\frac1Q\sum_{t=0}^{Q-1}(\Gamma(t)-\hat\Gamma(t))
 -(\Gamma(u)-\hat\Gamma(u)),
\]
where both the average and $\Gamma(u)-\hat\Gamma(u)$ lie in $[0,1/K]$, so their difference has absolute value at most $1/K$.
\end{proof}

We next compute a compressed distribution $\hat\mu_i$ for each $Y_i$.
Fix an integer $K\ge1$ and start from $\hat\mu_0=\mu_0$.
At client $v_i$, we first compute an approximate distribution of $Y_i=(Y_{i-1}+D_{v_i})\bmod Q$ using $\hat\mu_{i-1}$:
\begin{equation}\label{eq:mod-convolution}
\tilde\mu_i(z)=
\sum_{j=1}^{s_{v_i}}p_{v_i,j}
\hat\mu_{i-1}((z-d_{v_i,j})\bmod Q).
\end{equation}
Since $\tilde\mu_i$ has at most $Ks_{v_i}$ support values, we further compress it using Lemma~\ref{lem:distribution-compression} to obtain $\hat\mu_i$ with at most $K$ support values.
Using these compressed distributions, we approximate the expected extra cost by
\begin{equation}\label{eq:initial-approx}
\hat R_0(u)=\EE{\ell(\bm{D})}
+2\sum_{i=1}^n b_i\sum_z\hat\mu_i(z)h_u(z).
\end{equation}

\begin{lemma}\label{lem:initial-error}
For every integer $K\ge1$ and initial stock $u$, the compressed distributions satisfy $|\sum_z(\mu_i(z)-\hat\mu_i(z))h_u(z)|\le i/K$ for $0\le i\le n$.
Moreover, $|R_0(u)-\hat R_0(u)|\le2nc(T)/K$.
\end{lemma}
\begin{proof}
Let $e_i(u)=\sum_z(\mu_i(z)-\hat\mu_i(z))h_u(z)$, with $e_0(u)=0$.
By inserting $\tilde\mu_i$, we obtain
\[
e_i(u)=
\sum_z(\mu_i(z)-\tilde\mu_i(z))h_u(z)
+\sum_z(\tilde\mu_i(z)-\hat\mu_i(z))h_u(z).
\]
The second error term is the compression error, whose absolute value is at most $1/K$ by Lemma~\ref{lem:distribution-compression}.

To bound the first error term, let $f(z)=(z-d_{v_i,j})\bmod Q$ for each fixed $j$.
By independence and \eqref{eq:mod-convolution}, 
\[
 \sum_z(\mu_i(z)-\tilde\mu_i(z))h_u(z)
 =\sum_{j=1}^{s_{v_i}}p_{v_i,j}\sum_z(\mu_{i-1}(f(z))-\hat\mu_{i-1}(f(z)))h_u(z).
\]
By the definition of $h_u$, we have $h_u(z)=h_{(u-d_{v_i,j})\bmod Q}(f(z))+h_u(d_{v_i,j}\bmod Q)$.
Note that $\sum_z\mu_{i-1}(f(z))=\sum_z\hat\mu_{i-1}(f(z))=1$.
Thus, the first error term equals $\sum_jp_{v_i,j}e_{i-1}((u-d_{v_i,j})\bmod Q)$, and its absolute value is at most $\max_u|e_{i-1}(u)|$.
Therefore, $\max_u|e_i(u)|\le\max_u|e_{i-1}(u)|+1/K$, which is at most $i/K$ since $e_0=0$.

Finally, by Lemma~\ref{lem:initial-cost} and \eqref{eq:initial-approx}, we obtain
\[
|R_0(u)-\hat R_0(u)|\le2\sum_{i=1}^n|b_i||e_i(u)|\le\frac2K\sum_{i=1}^n i|b_i|\le\frac{2n}{K}c(T),
\]
where the last inequality follows from $|b_i|\le c(v_i,v_{i+1})$ and $\sum_i|b_i|\le c(T)$ by the triangle inequality.
\end{proof}

\subsection{Selecting the Initial Stock}

We now minimize $\hat R_0$ instead of $R_0$.
For fixed $z>0$, the function $h_u(z)$ decreases by $1$ when $u$ increases from $z-1$ to $z$ and is otherwise constant.
Thus, by \eqref{eq:initial-approx}, $\hat R_0$ changes only at positive support values of $\hat\mu_i$.
We select the smallest stock $u_K$ that minimizes $\hat R_0$ over $\lrc{0}\cup\bigcup_{i=1}^n\operatorname{supp}(\hat\mu_i)$, which contains at most $nK+1$ candidates.

\begin{lemma}\label{lem:initial-stock}
For every integer $K\ge1$, the selected initial stock satisfies
$R_0(u_K)\le\EE{\ell(\bm{D})}+2nc(T)/K$.
\end{lemma}
\begin{proof}
For each $z\in\lrc{0,\ldots,Q-1}$, exactly $z$ initial stocks satisfy $z>u$, so $\sum_{u=0}^{Q-1}h_u(z)=0$ by definition.
Taking the average in \eqref{eq:initial-approx} yields $(1/Q)\sum_{u=0}^{Q-1}\hat R_0(u)=\EE{\ell(\bm{D})}$.
Since $u_K$ minimizes $\hat R_0$, we obtain $\hat R_0(u_K)\le\EE{\ell(\bm{D})}$.
Moreover, by Lemma~\ref{lem:initial-error}, $R_0(u_K)\le\hat R_0(u_K)+2nc(T)/K$, which proves the lemma.
\end{proof}

We now combine the selected initial stock with the TSP tour from Lemma~\ref{lem:geometry}.

\begin{theorem}\label{thm:deterministic-priori}
For any $0<\varepsilon\le1$, the VRPSD admits deterministic a priori $(2+\varepsilon)$-approximation and $(3+\varepsilon)$-approximation algorithms for splittable and unsplittable deliveries, respectively, with running time $\O(n^7+(Sn/\varepsilon)\log(Sn/\varepsilon))$.
\end{theorem}
\begin{proof}
Set $K=\ceil{4n/\varepsilon}$, and use the tour $T$ from Lemma~\ref{lem:geometry} and the selected stock $u_K$ in both a priori policies.
By Lemma~\ref{lem:initial-stock},
$R_0(u_K)\le\EE{\ell(\bm{D})}+(\varepsilon/2)c(T)$.
Thus, their expected travel costs are at most
$(1+\varepsilon/2)c(T)+\EE{\ell(\bm{D})}$
and
$(1+\varepsilon)c(T)+2\EE{\ell(\bm{D})}$,
respectively.
By Lemmas~\ref{lem:benchmarks} and~\ref{lem:geometry}, $(1+\varepsilon/2)c(T)+\EE{\ell(\bm{D})}\le(2+\varepsilon)\EE{\Cs(\bm{D})}\le(2+\varepsilon)\Ops$ and $(1+\varepsilon)c(T)+2\EE{\ell(\bm{D})}\le(2+2\varepsilon)\EE{\Cs(\bm{D})}
+(1-\varepsilon)\EE{\ell(\bm{D})}\le(3+\varepsilon)\EE{\Cu(\bm{D})}
\le(3+\varepsilon)\Opu$.
Thus, we obtain the claimed approximation ratios.

Next, we analyze the running time.
At client $v_i$, we compute at most $Ks_{v_i}$ values $(z+d_{v_i,j})\bmod Q$, each with probability $\hat\mu_{i-1}(z)p_{v_i,j}$.
Sorting these values and combining equal values computes $\tilde\mu_i$, and then one scan computes $\hat\mu_i$.
Summing over all clients takes $\O(KS\log(KS))$ time.

To select $u_K$, we first compute $\hat R_0(0)$ in $\O(S+nK)$ time.
By \eqref{eq:initial-approx}, sorting the at most $nK$ support values and updating $\hat R_0$ at these values finds the smallest stock with the minimum value in $\O(nK\log(nK))$ time.
Since $S\ge n$ and $K=\O(n/\varepsilon)$, including the tour construction in Lemma~\ref{lem:geometry} yields the stated running time.
\end{proof}



\section{Lower Bounds on Approximation Ratios}\label{sec:lower}
In this section, we show that no a priori or restocking policy can guarantee an approximation factor below $2$ for splittable deliveries or below $3$ for unsplittable deliveries.
We use the same instance for both variants. 

\subsection{The Stochastic Instance}\label{sec:known-lower}
Fix an integer $m\ge2$, and set $N=m^4$, $p=m^{-3}$, $Q=2m^8$, and $\eta=m^{-8}$.
Consider the graph in Figure~\ref{fig:hard-instance}(a), with a depot $r$, bridge clients $x_1,\ldots,x_{2m}$, peak clients $y_1,\ldots,y_N$, and a non-client vertex $z$.
The edges $rx_j$ and $zx_j$ have cost $1/2$, and each edge $zy_i$ has cost $\eta/2$.
The metric is defined by minimum-cost paths.
For instance, the cost between distinct bridges is $1$ and between distinct peaks is $\eta$, and it holds that $c(r,y_i)=1+\eta/2$ and $c(x_j,y_i)=1/2+\eta/2$. 
Moreover, every bridge has demand $1$, and each peak has demand
\begin{equation}\label{eq:two-point}
 D_{y_i}=\begin{cases}
 Q-2,&\text{with probability }p,\\
 4,&\text{with probability }1-p.
 \end{cases}
\end{equation}
We call a peak \emph{large} or \emph{small} according to its realized demand.

\begin{figure}[t]
\centering
\caption{The stochastic instance and an adaptive depot tour. (a) The graph defining the metric, with demand $1$ at each bridge and demand $Q-2$ or $4$ at each peak. (b) A tour serving one large peak and two bridges, with total demand $Q$ and cost $2+\eta$.}\label{fig:hard-instance}
\vspace{-6pt}
\begin{subfigure}[t]{0.49\linewidth}
\centering
\caption{The bridge--peak metric}
\label{fig:bridge-peak-instance-a}
\begin{tikzpicture}[x=1.12cm,y=0.85cm,
  every node/.style={font=\small},
  client/.style={circle,draw,fill=white,inner sep=0pt,minimum size=5pt},
  depot/.style={rectangle,draw,fill=black,inner sep=0pt,minimum size=5pt},
  edge/.style={line width=0.55pt},
  route/.style={-{Stealth[length=4pt,width=3pt]},line width=0.75pt}]
\node[depot,label=left:$r$] (r) at (0,0) {};
\node[client,label=above:$x_1$] (x1) at (1.5,1.15) {};
\node[client,label=above:$x_2$] (x2) at (1.5,0.35) {};
\node at (1.5,-0.32) {$\vdots$};
\node[client,label=below:$x_{2m}$] (xm) at (1.5,-1.15) {};
\node[circle,draw,fill=black!15,inner sep=0pt,minimum size=5pt,
      label=below:$z$] (z) at (3,0) {};
\node[client,label=right:$y_1$] (y1) at (4.55,1.15) {};
\node[client,label=right:$y_2$] (y2) at (4.55,0.35) {};
\node at (4.55,-0.32) {$\vdots$};
\node[client,label=right:$y_N$] (yn) at (4.55,-1.15) {};
\draw[edge] (r)--node[above left,inner sep=1pt] {$\tfrac12$} (x1)
  --node[above right,inner sep=1pt] {$\tfrac12$} (z);
\draw[edge] (r)--(x2)--(z) (r)--(xm)--(z);
\draw[edge] (z)--node[above left,inner sep=1pt] {$\tfrac\eta2$} (y1);
\draw[edge] (z)--(y2) (z)--(yn);
\end{tikzpicture}
\end{subfigure}\hfill
\begin{subfigure}[t]{0.49\linewidth}
\centering
\caption{An adaptive delivery route}
\label{fig:bridge-peak-instance-b}
\begin{tikzpicture}[x=1.12cm,y=0.85cm,
  every node/.style={font=\small},
  client/.style={circle,draw,fill=white,inner sep=0pt,minimum size=5pt},
  depot/.style={rectangle,draw,fill=black,inner sep=0pt,minimum size=5pt},
  edge/.style={line width=0.55pt},
  route/.style={-{Stealth[length=4pt,width=3pt]},line width=0.75pt}]
\node[depot,label=left:$r$] (rr) at (0,0) {};
\node[client,label=above:$x_{2j-1}$] (xa) at (1.7,1.15) {};
\node[client,label=below:$x_{2j}$] (xb) at (1.7,-1.15) {};
\node[circle,draw,fill=black,inner sep=0pt,minimum size=5pt,
      label=right:{$y$ (large)}] (yy) at (3.65,0) {};
\draw[route] (rr)--node[above left,inner sep=1pt] {$\tfrac12$} (xa);
\draw[route] (xa)--node[above right,inner sep=1pt] {$\tfrac12+\tfrac\eta2$} (yy);
\draw[route] (yy)--node[below right,inner sep=1pt] {$\tfrac12+\tfrac\eta2$} (xb);
\draw[route] (xb)--node[below left,inner sep=1pt] {$\tfrac12$} (rr);
\node at (1.75,0.2) {demand $Q$};
\node at (1.75,-0.2) {cost $2+\eta$};
\end{tikzpicture}
\end{subfigure}
\end{figure}
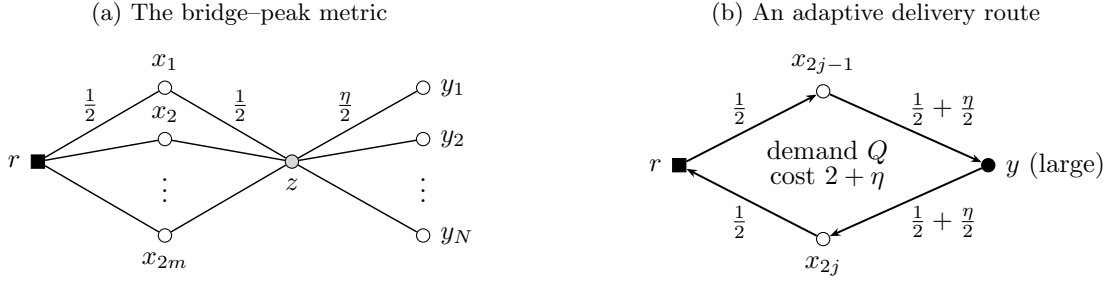

We next bound the expected cost of an adaptive policy for comparison.

\subsection{An Adaptive Comparison Policy}
We first visit all peaks on a depot tour of cost $2+N\eta$, serve the small peaks, and defer service to the large peaks.
This is feasible since the total small demand is at most $4N<Q$.

Let $K$ be the number of large peaks.
We then pair each of $\min\lrc{K,m}$ large peaks with two unused bridges and serve them on a depot tour $r,x_{2j-1},y,x_{2j},r$ of stock $Q$ and cost $2+\eta$, as shown in Figure~\ref{fig:hard-instance}(b).
Next, we serve each remaining large peak alone on a depot tour of cost $2+\eta$.
If $K<m$, the $2(m-K)$ remaining bridges can be served on one depot tour of cost $2(m-K)$, since their total demand is at most $2m<Q$.
Thus, the policy is feasible for both variants and has total cost
\begin{equation}\label{eq:adaptive-cost}
(2+N\eta)+(2+\eta)K+2(m-K)\ind_{K<m} =  2+2\max\lrc{K,m}+(N+K)\eta.
\end{equation}

We next bound the expected cost.
By definition, $\EE{K}=Np=m$ and $\operatorname{Var}(K)=Np(1-p)=m(1-p)$.
Then, $2\EE{\max\lrc{K-m,0}}=\EE{\size{K-m}}\le\sqrt{m(1-p)}$.
By \eqref{eq:adaptive-cost} and $\max\lrc{K,m}=m+\max\lrc{K-m,0}$, 
\begin{equation}\label{eq:adaptive-mean}
 \OPT_\nu\le2m+\sqrt{m(1-p)}+2+(N+m)\eta.
\end{equation}

\subsection{Splittable Deliveries}\label{sec:split-lower}
We first establish the lower bound for splittable a priori and restocking policies.
Fix a client order $\sigma$ chosen before observing any demand, and let $\cost_{\spv}$ denote the travel cost of a policy following $\sigma$.

To lower bound this cost, we expand each metric edge into a minimum-cost path, retain the original service locations, and contract the edges $zy_i$.
Since every peak has positive demand, each edge $zy_i$ is traversed at least twice, so the original cost is at least the contracted cost plus $N\eta$.
In the contracted graph, all peaks are at $z$.
The cost between $r$ and $z$ is $1$, from each bridge to $r$ or $z$ is $1/2$, and between distinct bridges is $1$.

We select each large peak if both neighbors are peaks and no other large peak occurs within two positions on either side.
Each selected peak defines a small--large--small block, and these blocks are disjoint.
Let $Z_{\spv}$ denote the number of selected blocks.
Note that at most $4m+2$ peaks are either at the beginning or end of $\sigma$ or adjacent to a bridge.
For every remaining peak, both neighbors are peaks.
Conditional on being large, it is not selected only if another large peak occurs within two positions on either side, which has probability at most $4p$ by a union bound.
Since each peak is large with probability $p$, by linearity of expectation,
\begin{equation}\label{eq:expected-good-split}
\mathbb E[Z_{\spv}\mid\sigma]\ge (N-4m-2)p(1-4p)\ge Np-(4m+2)p-4Np^2=m-\frac{8m^2+2m}{N}.
\end{equation}

\begin{lemma}\label{lem:ordered-split}
For every demand vector, a feasible splittable routing that follows $\sigma$ and serves each client before first visiting the next has cost at least $2m+2Z_{\spv}+N\eta$.
\end{lemma}
\begin{proof}
Each selected block has total demand $4+(Q-2)+4=Q+6>Q$.
Consider the part of the walk from the first delivery at its left small peak until its right small peak is fully served.
The stock at its start is at most $Q$, so the vehicle must visit the depot during this interval.
After contraction, this part of the walk starts and ends at $z$ and visits $r$, so its cost is at least $2$.

These parts of the walk are disjoint and contain no bridge service.
Removing them leaves a closed depot walk that still visits all $2m$ bridges at their service locations.
By shortcutting to these locations, its cost is at least $1/2+(2m-1)+1/2=2m$.
Thus, adding the removed costs and $N\eta$ proves the claimed bound.
\end{proof}

By Lemma~\ref{lem:ordered-split} and \eqref{eq:expected-good-split}, every splittable policy following $\sigma$ satisfies
\begin{equation}\label{eq:known-lower-split}
 \mathbb E[\cost_{\spv}\mid\sigma]\ge4m-\frac{16m^2+4m}{N}+N\eta.
\end{equation}
Since the bound holds for every $\sigma$, it also holds after averaging over a randomized client order.

\subsection{Unsplittable Deliveries}\label{sec:unsplit-lower}
We next establish the lower bound for unsplittable a priori and restocking policies.
We use the same client order $\sigma$ and the contracted graph from the previous subsection.
Let $\cost_{\unv}$ be the travel cost of such a policy.

We call a large peak \emph{isolated} if all its neighbors in $\sigma$ are small peaks, including when it is at an endpoint of $\sigma$.
Let $Z_{\unv}$ be the number of isolated large peaks.
At most $4m$ peaks are adjacent to bridges.
Among the at most $N-1$ pairs of consecutive peaks, each pair consists of two large peaks with probability $p^2$, preventing both from being isolated.
By linearity of expectation,
\begin{equation}\label{eq:expected-good-unsplit}
\mathbb E[Z_{\unv}\mid\sigma]\ge (N-4m)p-2(N-1)p^2\ge Np-4mp-2Np^2=m-\frac{6m^2}{N}.
\end{equation}

\begin{lemma}\label{lem:ordered-unsplit}
For every demand vector, a feasible unsplittable routing that follows $\sigma$ and serves each client before first visiting the next has cost at least $2m+4Z_{\unv}-2+N\eta$.
\end{lemma}
\begin{proof}
We first discard depot tours that serve no client.
The clients served on each remaining depot tour occupy consecutive positions in $\sigma$.
If such a tour serves an isolated large peak and another client, it must also serve a neighboring small peak, whose combined demand is $(Q-2)+4>Q$.
Thus, each isolated large peak requires a depot tour serving no other client, with cost at least $2$ in the contracted graph.
These tours have total cost at least $2Z_{\unv}$.

Next, suppose $Z_{\unv}\ge1$.
The isolated large peaks divide the remaining clients into $Z_{\unv}+1$ segments, with possibly empty end segments.
A depot tour cannot serve clients from different segments.
For each segment, we join its depot tours at $r$ and shortcut to the service locations in their order.
Each of the $Z_{\unv}-1$ middle segments starts and ends with a small peak, so the first and last edges of its resulting tour have total cost $2$.
Each edge costs at least $1/2$ times its number of bridge endpoints.
Since each bridge has two incident edges, the remaining edges have total cost at least $b$ for a middle segment with $b$ bridges.
Thus, its cost is at least $2+b$.
Similarly, each end segment has cost at least its number of bridges.

Since all $2m$ bridges lie in these segments, their total cost together with the isolated-peak tours is at least $2(Z_{\unv}-1)+2m+2Z_{\unv}=2m+4Z_{\unv}-2$.
If $Z_{\unv}=0$, a closed depot walk visiting all bridges has cost at least $2m$.
Thus, adding $N\eta$ proves the claimed bound in both cases.
\end{proof}

By Lemma~\ref{lem:ordered-unsplit} and \eqref{eq:expected-good-unsplit}, every unsplittable policy following $\sigma$ satisfies
\begin{equation}\label{eq:known-lower-unsplit}
 \mathbb E[\cost_{\unv}\mid\sigma]\ge6m-2-\frac{24m^2}{N}+N\eta.
\end{equation}
Since this bound holds for every $\sigma$, it also holds after averaging over a randomized client order.

We now compare both lower bounds with the adaptive upper bound.

\begin{theorem}\label{thm:lower}
For every $\varepsilon>0$, there exists an instance on which every a priori or restocking policy has an approximation ratio relative to the adaptive optimum greater than $2-\varepsilon$ for splittable deliveries and $3-\varepsilon$ for unsplittable deliveries.
This holds even with randomized ordering and replenishment, and full knowledge of demands after the order is fixed.

\end{theorem}
\begin{proof}
Since $N=m^4$ and $\eta=m^{-8}$, the ratios of the lower bounds in \eqref{eq:known-lower-split} and \eqref{eq:known-lower-unsplit} to the upper bound in \eqref{eq:adaptive-mean} tend to $2$ and $3$, respectively.
Both lower bounds hold for every client order and every choice of replenishments, even if all demands are known after the order is fixed.
Thus, choosing $m$ sufficiently large proves both claims on the same instance.
\end{proof}

\section{Applications}\label{sec:applications}

Our matroid-parity construction is independent of the demand distributions.
We apply it to several VRP variants with deterministic demands.
For each variant, let $\OPT$ denote its optimal solution value.

\subsection{Capacitated Vehicle Routing}\label{sec:app-cvrp}

The CVRP is a special case of the VRPSD with deterministic positive demands.
Thus, we immediately obtain randomized $2$-approximation and $3$-approximation algorithms for splittable and unsplittable deliveries, respectively.
These algorithms can be derandomized in polynomial time by considering the $\O(n)$ distinct shortage patterns induced by the initial stock $U$.

For unsplittable deliveries, given a TSP tour $T$, \citet{journal:friggstad:MOR:25} proposed an LP-based algorithm that computes in $n^{\O(1/\varepsilon)}$ time a solution of cost at most $c(T)+(1+\varepsilon)\ell(\bm{q})+\ln2\OPT$.
By Lemmas~\ref{lem:benchmarks} and~\ref{lem:geometry}, our tour satisfies $c(T)\le2\OPT-\ell(\bm{q})$.
Since $\ell(\bm{q})\le\OPT$, we obtain an approximation ratio of $2+\ln2+\varepsilon$.

\begin{corollary}\label{cor:cvrp}
For any constant $\varepsilon>0$, the CVRP admits $2$-approximation and $(2+\ln2+\varepsilon)$-approximation algorithms for splittable and unsplittable deliveries, respectively.
\end{corollary}

These results improve the previous approximation ratios of about $2.5-1/3000$~\citep{journal:blauth:MP:23} and $3.159$~\citep{journal:ebert:arXiv:26}, respectively.

\begin{remark}
The above results are direct applications of our tour construction. This construction can also be combined with more refined algorithms, such as that of \citet{journal:zhao:arXiv:26}, to further improve the approximation ratio for the unsplittable CVRP. 
We do not pursue such improvements here.
\end{remark}

\subsection{Multi-Depot Capacitated Vehicle Routing}
\label{sec:app-mdcvrp}
The multi-depot VRP (MD-VRP) extends the CVRP to multiple depots $D$, with input graph $G=(V\cup D,E)$.
The objective is to find a minimum-cost set of depot tours that serve all clients, with each tour starting and ending at the same depot.

By contracting all depots into a single depot $r$ and taking the metric closure, we obtain a graph $G/D$ with metric cost function $c_D$.
Let $\OPT_D$ denote the optimal cost of the corresponding CVRP instance on $G/D$, and define $\ell_D(\bm{q})=(2/Q)\sum_{v\in V}c_D(r,v)q_v$.
It holds that $\ell_D(\bm{q})\le\OPT_D\le\OPT$~\citep{journal:li:IJOC:90}.

Given a TSP tour $T$ in $G/D$, the tour-partition algorithm of \citet{journal:li:IJOC:90} computes in polynomial time a solution of cost at most $2c_D(T)+\ell_D(\bm{q})$ for splittable deliveries and $2c_D(T)+2\ell_D(\bm{q})$ for unsplittable deliveries.
The tree-partition algorithm of \citet{journal:harks:TS:13} computes in polynomial time a solution of cost at most $2\OPT_D+2\ell_D(\bm{q})$ for both variants.

By Lemmas~\ref{lem:benchmarks} and~\ref{lem:geometry}, our TSP tour in $G/D$ satisfies $c_D(T)\le 2\OPT_D-\ell_D(\bm{q})$.

\paragraph{Splittable case.}
Running the two partition algorithms yields a solution of cost at most
\[
\min\lrC{4\OPT_D-\ell_D(\bm{q}),
2\OPT_D+2\ell_D(\bm{q})}
\le\frac{10}{3}\OPT_D.
\]
Since $\OPT_D\le\OPT$, we obtain a $10/3$-approximation algorithm for the splittable MD-VRP.

\paragraph{Unsplittable case.}
By applying the LP-rounding method of \citet{journal:friggstad:MOR:25} to the tour-partition algorithm of \citet{journal:li:IJOC:90} with any fixed $0<\gamma\le\ln2$, \citet[Theorem~16]{journal:zhao:TCS:25} implicitly proved that the resulting LP-based tour-partition algorithm has cost at most 
$
2c_D(T)+(2e^{-\gamma}+\varepsilon)\ell_D(\bm{q})+\gamma\OPT.
$
Combining it with the tree-partition algorithm yields a solution of cost at most
\[
\min\lrC{
(4+\gamma)\OPT+
(2e^{-\gamma}-2+\varepsilon)\ell_D(\bm{q}),
2\OPT+2\ell_D(\bm{q})
}.
\]
Using $\ell_D(\bm{q})\le\OPT$ and choosing $\gamma$ as the positive root of $2e^\gamma=3+\gamma$, we obtain an approximation ratio of at most $3.792$ for sufficiently small $\varepsilon$.

\begin{corollary}\label{cor:md-unsplit}
The MD-VRP admits $10/3$-approximation and $3.792$-approximation algorithms for splittable and unsplittable deliveries, respectively.
\end{corollary}

These results improve the previous approximation ratios of $3.9365$~\citep{conference:friggstad:IPCO:26} for unit demands, and about $4-1/1500$ and $4-1/50000$ for splittable and unsplittable deliveries, respectively~\citep{journal:zhao:TCS:25}.

\subsection{Regret-Bounded Vehicle Routing and School Bus Problem}
\label{sec:app-sbp}

The (additive) regret-bounded VRP (RVRP) is defined on the same metric graph as the CVRP.
Given $R\ge0$, the objective is to find a minimum number of paths starting at $r$ and covering all clients, such that the travel cost from $r$ to each client $v$ along its path is at most $\rho_v+R$.
For a path $P$ ending at $t$, its \emph{regret} is $c(P)-\rho_t$.
By the triangle inequality, regret is nondecreasing along the path, so bounding the regret of $P$ by $R$ ensures the constraint for every client.
If each route is additionally required to serve at most $Q$ clients, the problem is known as the school bus problem (SBP)~\citep{journal:bock:ALG:13}.

For $R=0$, the RVRP can be solved in polynomial time~\citep{conference:friggstad:STOC:14}.
For $R>0$, we modify the WGMP instance from Section~\ref{sec:geometry} by assigning each marker line for client $v$ weight $2R-4\rho_v$.
We compute a minimum-weight parity basis of weight $W$.
As in Lemma~\ref{lem:pairing}, each component of the forest formed by the selected client lines has one depot spoke in $G_1$ and one marker endpoint.
We keep the path from $r$ to this endpoint once and double the remaining edges.
If $p$ marker lines are selected, then shortcutting yields $m\le2p$ paths covering $V$, without increasing their total regret.
Let $\Delta_{\mathrm{reg}}$ denote the total regret of these paths.
Since each path from $r$ to its marker endpoint $v$ has cost at least $\rho_v$ by the triangle inequality, we obtain
\[
\Delta_{\mathrm{reg}}
\le 2c(F\cup E_{\mathrm{d}})
-4\sum_{v\in V_{\mathrm{m}}}\rho_v
=W-2pR
\quad\Longrightarrow\quad
mR+\Delta_{\mathrm{reg}}\le W.
\]

To bound $W$, consider an optimal RVRP solution with $\OPT$ paths.
We may assume that each client appears on exactly one path by shortcutting repeated clients.
For each path $P$ ending at $t$, we construct the corresponding lines using $P$ and the direct edge $rt$, as in Lemma~\ref{lem:embedding}.
Their total weight is
\[
2(c(P)+\rho_t)+2R-4\rho_t
=2(c(P)-\rho_t)+2R\le4R.
\]
Thus, $W\le4R\OPT$ and $mR+\Delta_{\mathrm{reg}}\le4R\OPT$.
\citet[Lemma~2.2]{conference:friggstad:STOC:14} showed that these paths can be split into at most $m+\Delta_{\mathrm{reg}}/R\le4\OPT$ paths, each with regret at most $R$.
Thus, we obtain a $4$-approximation algorithm for the RVRP.
Moreover, \citet[Lemma~1]{journal:bock:ALG:13} showed that a $\rho$-approximation algorithm for the RVRP yields a $(\rho+1)$-approximation algorithm for the SBP.

\begin{corollary}\label{cor:school-bus}
The RVRP and SBP admit $4$-approximation and $5$-approximation algorithms, respectively.
\end{corollary}

These results improve the previous approximation ratios of $10$ and $11$ for the RVRP and SBP, respectively~\citep{journal:blauth:MP:23}.

\subsection{Capacitated Arc Routing}\label{sec:app-carp}

The capacitated arc routing problem (CARP) is defined on a complete graph $G=(V\cup\lrc{r},E)$ with required edges $E^+$, where each $e\in E^+$ has demand $q_e\in\lrc{1,\ldots,Q}$.
The objective is to find a minimum-cost set of depot tours serving all required-edge demands, with each tour serving total demand at most $Q$.
Each required edge must be traversed and its entire demand served by one tour. In the unit-demand variant, $q_e=1$ for every $e\in E^+$.
We assume metric edge costs, since \citet{journal:vanbevern:ORL:14} showed that any approximation ratio for the metric CARP also holds for the general CARP.
By taking separate copies of shared required-edge endpoints and removing vertices not incident to any required edge, we may assume that $E^+$ is a perfect matching on $V$.
We also let $n=\size{V}$. Then, $\size{E^+}=n/2$.

Let $\bm{q}=(q_e)_{e\in E^+}$ and $\rho_v=c(r,v)$.
For each $e=uv$, define $\eta_e=(\rho_u+c(e)+\rho_v)/2$ and $\ell_E(\bm{q})=(2/Q)\sum_{e\in E^+}q_e\eta_e$.
It holds that $\ell_E(\bm{q})\le\OPT$~\citep{journal:jansen:NETWORKS:93}.

We modify the WGMP instance from Section~\ref{sec:geometry}, using two vertex-disjoint complete graphs $G_1,G_2$ with vertex labels $V\cup\lrc{r_1}$ and $V\cup\lrc{r_2}$, respectively.
Each depot spoke in $G_1$ has $n-1$ parallel copies.
Client lines pair the corresponding edges $xy$ in the two graphs, with weight $2c(x,y)$.
Depot lines pair the spokes $r_1x,r_1y$ in $G_1$, with weight $2\rho_x+2\rho_y$.
For each $e=uv\in E^+$, its marker line pairs $r_2u,r_2v$ and has weight $2c(e)-2\eta_e$.

The client line for $e$ and its marker line cannot both be selected, since their edges form a cycle in $G_2$.
On the other hand, selecting all marker lines and the depot lines pairing required-edge endpoints forms a parity basis.
We seek a minimum-weight parity basis that selects exactly one of the client line and marker line for each required edge.
This can be achieved by reducing the weights of these two lines for each required edge by the same constant, larger than the sum of the absolute values of all original line weights.
The weight reduction first maximizes the number of selected lines corresponding to required edges and then minimizes the original total weight.
Since at most one line can be selected for each required edge and a parity basis selecting all marker lines exists, the resulting parity basis selects exactly one of these lines for each required edge and minimizes the original total weight among all such bases.
Let $W$ denote its original total weight.

Let $F$ be the forest formed by the selected client lines, and let $E_{\mathrm{d}}$ be the selected depot spokes in $G_1$.
As in Lemmas~\ref{lem:pairing} and~\ref{lem:decoding}, for each required edge $e=uv$ whose marker line is selected, we add $e$ between its two forest components to form a depot cycle $Z_e$.
We keep each cycle once and double the remaining edges.
Since $c(Z_e)\ge\rho_u+c(e)+\rho_v=2\eta_e$ by the triangle inequality, the resulting multigraph has cost at most $W$.
Every required edge is included, either in the forest or on an added cycle.
Taking an Eulerian tour and shortcutting while retaining each required edge yields an rural postman problem (RPP) tour $T$ with $c(T)\le W$.

To bound $W$, we apply the matching argument of Lemma~\ref{lem:allocation}, with required edges in place of clients and weights $\eta_e$.
This gives a partition $\P$ of $E^+$ and depot tours $T_B$ traversing each assigned edge once.
For each $B\in\P$, choose $e_B\in B$ maximizing $\eta_e$.
Then,
\[
\sum_{B\in\P}c(T_B)\le\OPT,
\qquad
2\sum_{B\in\P}\eta_{e_B}\ge\ell_E(\bm{q}).
\]
For each block $B$, removing $e_B$ from $T_B$ leaves two depot-rooted paths.
As in Lemma~\ref{lem:embedding}, their client and depot lines, together with the marker line for $e_B$, form a parity basis.
This basis also selects exactly one of the client line and marker line for each required edge, so both bases have the same weight reduction.
Hence,
\[
c(T)\le W
\le\sum_{B\in\P}\lra{2c(T_B)-2\eta_{e_B}}
\le2\OPT-\ell_E(\bm{q}).
\]

Given a RPP tour $T$, the tour-partition algorithm of \citet{journal:jansen:NETWORKS:93} computes in polynomial time a solution of cost at most $c(T)+\ell_E(\bm{q})$ for the unit-demand variant and $c(T)+2\ell_E(\bm{q})$ for the unsplittable variant.

Since our tour satisfies $c(T)\le2\OPT-\ell_E(\bm{q})$ and $\ell_E(\bm{q})\le\OPT$, we obtain the following result.

\begin{corollary}\label{cor:carp}
The CARP admits $2$-approximation and $3$-approximation algorithms for the unit-demand and unsplittable variants, respectively.
\end{corollary}

These results improve the previous approximation ratios of $2.5$ and $3.5$ for the unit-demand and unsplittable variants, respectively~\citep{journal:jansen:NETWORKS:93}.
For unsplittable deliveries, we further apply the LP-rounding method of \citet{journal:friggstad:MOR:25} to the tour-partition algorithm of \citet{journal:jansen:NETWORKS:93}.
We claim that this yields a solution of cost at most $c(T)+(1+\varepsilon)\ell_E(\bm{q})+\ln2\cdot\OPT$, improving the approximation ratio to $2+\ln2+\varepsilon$.

\bibliographystyle{apalike}
\bibliography{VRPSD}
\end{document}